\documentclass{IEEEtran}
\usepackage{cite}
\usepackage{amsmath,amssymb,amsfonts}
\usepackage{algorithmic}
\usepackage{graphicx}
\usepackage{textcomp}
\usepackage{hyperref, comment}
\usepackage{textcomp}
\usepackage{xcolor}
\usepackage{dsfont}

\newtheorem{lemma}{Lemma}
\newtheorem{theorem}{Theorem}
\newtheorem{corollary}{Corollary}
\newtheorem{remark}{Remark}
\newtheorem{definition}{Definition}
\newtheorem{example}{Example}

\newtheorem{problem}{Problem}
\newtheorem{condition}{Condition}

\usepackage{dblfloatfix}
\def\BibTeX{{\rm B\kern-.05em{\sc i\kern-.025em b}\kern-.08em
    T\kern-.1667em\lower.7ex\hbox{E}\kern-.125emX}}
\begin{document}
\title{Robust Cooperative Output Regulation of Discrete-Time  Heterogeneous Multi-Agent Systems}

\author{K\"{u}r\c{s}ad Metehan G\"{u}l and Selahattin Burak Sars\i lmaz
\thanks{K\"{u}r\c{s}ad Metehan G\"{u}l and Selahattin Burak Sars\i lmaz are with the Department of Electrical and Computer Engineering, Utah State University, Logan, UT 84322, USA (emails: 
        {\tt kursad.gul@usu.edu, burak.sarsilmaz@usu.edu}).}%
}

\maketitle
\begin{abstract}


This article considers robust cooperative output regulation of discrete-time uncertain heterogeneous (in dimension) multi-agent systems (MASs). We show that the solvability of this problem with an internal model-based distributed control law reduces to the existence of a structured control gain that makes the nominal closed-loop system matrix of the MAS Schur. Accordingly, this article focuses on global and agent-wise local sufficient conditions for the existence and design of such a structured control gain. Based on a structured Lyapunov inequality, we present a convexification that yields a linear matrix inequality (LMI), whose feasibility is a global sufficient condition for the existence and design.  Considering the individual nominal dynamics of each agent, the existence is also ensured if each agent solves a structure-free control problem. Its convexification yields LMIs that allow each agent to separately design its structure-free control gain.  
Lastly, we study the relationships between the sets of control gains emerging from both global and local perspectives.

\end{abstract}


\section{Introduction}
\subsection{Literature Review}

The extension of the robust output regulation problem \cite{1101137, FRANCIS1976457, huang2004nonlinearbook} to multi-agent systems (MASs), known as the robust cooperative output regulation problem (RCORP), provides a unifying framework to formulate and solve cooperative control problems, including leader-following consensus \cite{cai2022cooperative}, formation tracking \cite{scitechpaper}, and containment \cite{10438523}. For continuous-time homogeneous (in dimension) \cite{6361274,cai2022cooperative} and heterogeneous (in dimension) \cite{koru2020cooperative,kawamura2020distributed,SarsilmazIJC, koru2022regional}  linear MASs over general directed graphs, RCORP has been well studied by leveraging the distributed internal model approach. The synthesis of control gains for internal model-based distributed control of heterogeneous MASs can be grouped into two: global design \cite{koru2020cooperative} and agent-wise local design \cite{ kawamura2020distributed, SarsilmazIJC, koru2022regional}.  
A global design method treats the closed-loop MAS as a single system, imposing a particular structure on the overall control gain, especially zeros for numerous entries. 
Such structured control synthesis problems are believed to be NP-hard \cite{nphardness}. On the other hand, an agent-wise local design method focuses on the individual agent dynamics, allowing each agent to design its structure-free control gain independent of other agents.

One computationally efficient approach to address structured control problems is to restrict solutions of the Lyapunov inequalities to specific structures \cite{Stipanovic01012001,SILJAK2005169}. The existence of such structured solutions has been studied for a few  structures \cite{matrixdiagonalstabilitybook, ALEKSANDROV201638, 8876610}. Specifically, in \cite{ALEKSANDROV201638}, a necessary and sufficient condition for the existence of a diagonal solution to the algebraic Riccati equation is derived, whereas sufficient conditions for the existence of a block-diagonal solution to  Lyapunov and $\mathcal{H}_{\infty}$ Riccati inequalities are presented in \cite{8876610}. Though in general it remains an open question which is less conservative,  another computationally efficient way is to pass the structure-imposing to the new matrices through a characterization of the Lyapunov inequalities based on the projection lemma \cite{8747426}.

Compared to its continuous-time counterpart, RCORP for discrete-time linear MASs is seldom addressed. For homogeneous (in dimension) MASs, the nominal closed-loop system matrix is similar to a block upper or lower triangular matrix whose $i$th diagonal block consists of nominal system parameters of agent $i$ and an eigenvalue induced by the graph matrices \cite{7004551,homogenpcopyyaminhuang, Cai2020}. However, the nominal closed-loop system matrix does not enjoy this useful property when the MAS is heterogeneous in dimension. In this case, to establish agent-wise local design methods,  \cite{8360489} derives small-gain solvability conditions, and  \cite{7018013} employs a distributed observer that provides each agent with an estimate of the exogenous signal.

\subsection{Contribution}

This article considers the RCORP of discrete-time uncertain heterogeneous (in dimension)  linear time-invariant MASs over  general time-invariant directed graphs with an internal model-based distributed dynamic state feedback control law, which does not exchange controller states of neighboring agents. 
This control law solves the RCORP under standard conditions if there is a structured control gain, with a particular structure, that makes the nominal closed-loop system matrix Schur (see Theorem \ref{thm:keytheorem} and Remark  \ref{rmk:sufficient_necessary}). Though the stabilizability of the pair of interest, established in Lemma \ref{lmm:augsysstab}, is necessary for the existence of such a control gain, it is not sufficient in general (see Example \ref{ex_structured_stabilizability}).
Accordingly, this article presents global and agent-wise local sufficient conditions for the existence and design of the structured control gain. 
Based on a structured Lyapunov inequality, the feasibility of a linear matrix inequality (LMI) provides a global sufficient condition for the existence and design (see Lemma \ref{lmm:globaldesign}, Remark \ref{remark_stabilizability}, and Theorem \ref{thm:globaltheorem}). 
The existence is also ensured if each agent solves a structure-free control problem (see Theorem \ref{thm:agentwiseLocal_Cyclic}). Its convexification yields LMIs for each agent's control gain synthesis (see  Remark \ref{remark_convex_nonconvex_local} and Corollary \ref{crl:ControlSynthesisDiszero}). Therefore, the proposed global and agent-wise local design methods can be
conducted in polynomial time. 

In terms of scalability and conservatism, we briefly compare both design methods. The global design method simultaneously synthesizes each agent's control gain, whereas the agent-wise local design method does so separately. Thus, the agent-wise local design method is much more scalable.  Section 
\ref{sec_set_inclusions_global_local} investigates the relationships between the sets of control gains emerging from both perspectives. One of the results is that
the global design method is less conservative than the agent-wise local design method (see Corollary \ref{crl:setinclusion1} and Example \ref{ex:KsdoesnotimplyKlc}).

This article differs from the relevant studies \cite{8360489,7018013} in the following key aspects. The proof of the solvability result in \cite{8360489} enforces stability of the nominal local dynamics of each agent, disregarding the case in which the closed-loop system matrix of the MAS is Schur with unstable nominal local dynamics. In contrast, the solvability result in this article does not disregard such cases (see Theorem \ref{thm:keytheorem} and Example \ref{ex:KGdoesnotimplyKa}), and hence, it is more general. While \cite{8360489} provides an agent-wise local perspective on the existence and design of the structured control gain, this article establishes both global and agent-wise local sufficient conditions, with the local conditions implying the global one. In particular, the proposed agent-wise local sufficient conditions construct a Lyapunov inequality for the nominal closed-loop system matrix of the MAS, which is a useful tool for handling switching graphs \cite{10540292}, event-triggered control \cite{QIAN2024119}, and bounded parametric uncertainties within the context of the continuous-time RCORP.  Lastly, the distributed control law in  \cite{7018013} necessitates a communication network for the MAS since agents exchange their controller states, whereas the considered distributed control law in this article can still be applicable in the lack of a communication network since agents use their relative outputs with respect to their neighbors, which can be measured by their onboard sensors.

\subsection{Notation}
We write  $\mathbf{1}_n$ for the $n \times 1$ vector of all ones,  $I_n$ for the $n \times n$ identity matrix, $0_{n \times m}$ or $0$ for the $n \times m$ zero matrix, 
$\otimes$  for the Kronecker product, and $\mathrm{diag}(X_i)$  for the block-diagonal matrix whose entries on the main diagonal are $X_i$s for all $i$ in the corresponding index set. 
The vector formed by stacking columns of a matrix $X \in \mathbb{R}^{n \times m}$ is denoted by $\mathrm{vec}(X)$.  
The spectrum and determinant of a square matrix $X \in \mathbb{R}^{n \times n}$ are denoted by $\mathrm{spec}(X)$ and $\mathrm{det}(X)$, respectively. 
Let $\mathbb{S}^n_{++}$ denote the set of $n \times n$ symmetric positive definite matrices. 
For symmetric matrices $X$ and $Y$,  $X  \succ Y $ $(\succeq, \prec, \preceq)$ means $X-Y$ is positive definite (positive semidefinite, negative definite, negative semidefinite). We write  $\lambda_{\mathrm{min}}(\cdot)$ and  $\lambda_{\mathrm{max}}(\cdot)$ for the minimum and maximum eigenvalue of a symmetric matrix.
Lastly, the cardinality of a set $S$ is denoted by $|S|$.

\section{Problem Formulation}
We consider a discrete-time  uncertain heterogeneous MAS in the form 
\begin{align}\label{eq:MAS} \nonumber
x_i^{+} &= \Bar{A}_ix_i + \Bar{B}_iu_i + E_iv \\
e_i &= y_i - Fv = \Bar{C}_ix_i + \Bar{D}_iu_i   - Fv  , \quad i = 1, \ldots, N 
\end{align}
where $x_i \in \mathbb{R}^{n_i}$ is the state, $u_i \in \mathbb{R}^{m_i}$ is the control input, and $e_i \in \mathbb{R}^{p}$ is the tracking error of subsystem $i$. Moreover, $\Bar{A}_i = A_i + \delta A_i$, $\Bar{B}_i = B_i + \delta B_i$, $\Bar{C}_i = C_i + \delta C_i$, and $\Bar{D}_i = D_i + \delta D_i$ are uncertain matrices with $A_i, B_i, C_i$, and $D_i$ being their nominal parts. 
The results of this article hold for any  matrices $E_i$ and $F$, which do not need to be known.
Also,
$v \in \mathbb{R}^{n_0}$ is the exogenous signal generated by the  exosystem 
\begin{equation}\label{eq:exosystem}
    v^{+} = A_0v.
\end{equation}
This autonomous  system yields the reference $Fv$ to be tracked and the disturbance   $E_iv$ to be rejected by  
subsystem $i$.

In the context of the RCORP, the subsystems of \eqref{eq:MAS}, considered the followers, and the exosystem \eqref{eq:exosystem}, considered the leader, constitute a leader-follower MAS of $N+1$ agents. To model the information exchange between $N$ followers, we use a time-invariant directed graph ${\mathcal{G}} = (\mathcal{N}, \mathcal{E})$ without self-loops, where  $\mathcal{N} = \left \lbrace 1, \ldots, N \right \rbrace$ is the node set and $\mathcal{E}  \subseteq \left \lbrace (i,j)    ~|~ i, j \in \mathcal{N}, i \neq j  \right \rbrace $ is the edge set. Here, node $i \in \mathcal{N}$ corresponds to follower $i$, and for each $j,i \in \mathcal{N}$, we put $(j,i) \in \mathcal{E}$ if, and only if, follower $i$ has access to the information of follower $j$. The entries of the adjacency matrix $\mathcal{A} = [a_{ij}] \in \mathbb{R}^{N\times N}$ of the graph $\mathcal{G}$ are determined by the rule that for each $j, i \in \mathcal{N}$, $a_{ij} > 0 $  if $(j, i) \in \mathcal{E}$ and $a_{ij} = 0$ otherwise. 
The in-degree $d_i$ of follower $i$ is defined as  $d_{i} = \sum_{j=1}^{N} a_{ij}$. Furthermore, the leader is included in the information exchange model by augmenting the graph ${\mathcal{G}}$ as follows: 
Let $\Bar{\mathcal{G}} = (\Bar{\mathcal{N}}, \Bar{\mathcal{E}}) $ be a  directed graph with $\Bar{\mathcal{N}} = \mathcal{N} \cup \{0\}$, $\Bar{\mathcal{E}} = \mathcal{E} \cup \mathcal{E}'$, where $\mathcal{E}' \subsetneq \{(0,i) ~|~ i \in \mathcal{N}\}$. Here, node $0$ corresponds to the leader, and for any $i \in \mathcal{N}$, we put  $(0,i) \in \mathcal{E}'$ if, and only if, follower $i$ has access to the leader's information. For any $i\in \mathcal{N}$, the pinning gain $g_{i} > 0$ if $(0,i) \in \mathcal{E}'$ and $g_{i} = 0$ otherwise.
A control law that relies on the information exchange modeled by the augmented directed graph $\Bar{\mathcal{G}}$ is called a \textit{distributed control law}.

The tracking error $e_i$ is available to a proper subset of followers.
Each  follower also has access to the relative output  $y_i - y_j = e_i - e_j$ for some $j \in \mathcal{N}$, depending on $\mathcal{G}$. Based on the available information, a local virtual tracking error for each follower is defined as
\begin{align}\label{eq:virtualerror}
    e_{\mathrm{v}i} = \frac{1}{d_i+g_i}\left(\sum_{j=1}^{N} a_{ij}(e_i-e_j) + g_ie_i \right)
\end{align}
where we assume $d_i + g_i > 0$ for all $i \in \mathcal{N}$ in this article\footnote{
Condition \ref{ass:spanningtree} is sufficient for this assumption to hold, but it is not necessary in general. It becomes necessary when the directed graph $\mathcal{G}$ is acyclic.}. 
Consider the distributed dynamic state feedback control law
\begin{align}\label{eq:controllaw} \nonumber
    z_i^{+} &= G_{1i}z_i + G_{2i}e_{\mathrm{v}i}   \\
    u_i &= K_{1i}x_i + K_{2i}z_i, \quad i = 1, \ldots, N
\end{align}
where $z_i \in \mathbb{R}^{n_z}$ is the controller state of follower $i$. The matrices $G_{1i} \in \mathbb{R}^{n_z \times n_z}$,  $G_{2i} \in \mathbb{R}^{n_z \times p}$, $K_{1i}\in \mathbb{R}^{m_i \times n_i}$, and $K_{2i} \in \mathbb{R}^{m_i \times n_z}$ are  control gains to be designed.  

For convenience, the MAS uncertainty is represented with  
\begin{align}
    \delta = \begin{bmatrix}
        [\mathrm{vec}(\delta A_1)^{\mathrm{T}},\ldots,\mathrm{vec}(\delta A_N)^{\mathrm{T}}]^{\mathrm{T}} \\ [\mathrm{vec}(\delta B_1)^{\mathrm{T}},\ldots,\mathrm{vec}(\delta B_N)^{\mathrm{T}}]^{\mathrm{T}} \\ [\mathrm{vec}(\delta C_1)^{\mathrm{T}},\ldots,\mathrm{vec}(\delta C_N)^{\mathrm{T}}]^{\mathrm{T}} \\ [\mathrm{vec}(\delta D_1)^{\mathrm{T}},\ldots,\mathrm{vec}(\delta D_N)^{\mathrm{T}}]^{\mathrm{T}}
    \end{bmatrix}. 
    \nonumber
\end{align}
The \textit{closed-loop system} consists of \eqref{eq:MAS}, \eqref{eq:virtualerror}, and \eqref{eq:controllaw}. It is referred to as the $\textit{nominal closed-loop system}$ if $\delta = 0$. 
We now state the RCORP for discrete-time uncertain heterogeneous MASs. Its definitions for continuous-time and discrete-time MASs with identical nominal follower dynamics can be found in Definition 1 of  \cite{6361274} and Problem 1 of \cite{7004551}, respectively. 

\begin{problem}[RCORP]\label{prb:mainproblem}
    Given the MAS composed of \eqref{eq:MAS} and \eqref{eq:exosystem}, and the augmented directed graph $\Bar{\mathcal{G}}$, find a distributed control law of the form \eqref{eq:controllaw} such that
    \begin{enumerate}
        \item [(i)] The nominal closed-loop system matrix is Schur\footnote{
   That is, all its eigenvalues have modulus less than $1$.
       };
     \item[(ii)] For $i \in \mathcal{N}$,  for any $x_i(0)$, $z_i(0)$, $v(0)$, $E_i$, and $F$, and for any $\delta \in \Delta$, where
     $\Delta$ is a set of $\delta$ such that the closed-loop system matrix is Schur,      $ \lim_{t\to\infty} e_i(t) = 0$. 

    \end{enumerate}
\end{problem}

\begin{remark}\label{rmk:ExistenceofDelta}
It is well known that if the nominal closed-loop system matrix is Schur, then the set $\Delta$ contains $0$ in its interior.
Thus, only the nominal closed-loop system matrix is considered in property (i) of Problem \ref{prb:mainproblem}.  
\end{remark}

\section{Solvability of the RCORP}

For the solvability of Problem \ref{prb:mainproblem}, we will refer to the following conditions. 

\begin{condition}\label{ass:spanningtree}
The augmented directed graph $\Bar{\mathcal{G}}$ contains a directed spanning tree.
\end{condition}

\begin{condition}\label{ass:A0antiSchur}
For any $\lambda \in \mathrm{spec}(A_0)$, $|\lambda| \geq 1$.
\end{condition}

\begin{condition}\label{ass:pcopy}
For any $i \in \mathcal{N}$, the pair $(G_{1i},G_{2i})$ incorporates a $p$-copy internal model of $A_0$\footnote{The pair $(G_{1i}, G_{2i})$  takes the following form: $G_{1i} = \mathrm{diag}(\alpha_{\ell i})$ and $G_{2i} = \mathrm{diag}(\beta_{\ell i})$, where  for $\ell = 1, \ldots, p$,
$\alpha_{\ell i}$ is a square matrix and $\beta_{\ell i}$ is a column vector such that the minimal polynomial of $ A_0$ is equal to
the characteristic polynomial of $\alpha_{\ell i}$ and the pair $(\alpha_{\ell i}, \beta_{\ell i})$ is reachable.}.
\end{condition}

\begin{condition}\label{ass:transzero}
For any $i \in \mathcal{N}$, 
\begin{align}
    \mathrm{rank}\begin{bmatrix}
        A_i - \lambda I_{n_i} & B_i \\ C_i & D_i
    \end{bmatrix} = n_i + p \quad \forall{\lambda} \in \mathrm{spec}(A_0).
    \nonumber
\end{align}
\end{condition}

\begin{condition}\label{ass:AiBistab}
For any $i \in \mathcal{N}$, the pair $(A_{i},B_{i})$ is stabilizable.
\end{condition}

Conditions \ref{ass:A0antiSchur}--\ref{ass:AiBistab} are standard in linear  output regulation theory (e.g., see Chapter 1 in \cite{huang2004nonlinearbook}), and so is Condition \ref{ass:spanningtree} while tackling the RCORP (e.g., see \cite{6361274, koru2020cooperative}). 


Considering \eqref{eq:MAS}--\eqref{eq:controllaw}, the closed-loop system and the exosystem can be compactly written as  
\begin{align}\label{eq:composite_system}
    x_\mathrm{g}^{+} &= \Bar{A}_\mathrm{g}x_\mathrm{g} + B_\mathrm{g}v_\mathrm{a} \nonumber \\
        {v}_\mathrm{a}^{+} &= A_{0\mathrm{a}}v_\mathrm{a} \nonumber \\
    e &= \Bar{C}_\mathrm{g}x_\mathrm{g} + D_\mathrm{g}v_\mathrm{a}  
\end{align}
with  $x_\mathrm{g} = [x^\mathrm{T}, z^\mathrm{T}]^\mathrm{T}$, $x = [x_1^\mathrm{T},\ldots,x_N^\mathrm{T}]^\mathrm{T}$, $z = [z_1^\mathrm{T},\ldots,z_N^\mathrm{T}]^\mathrm{T}$, $v_\mathrm{a} = \mathbf{1}_N \otimes v$, $e = [e_1^\mathrm{T},\ldots,e_N^\mathrm{T}]^\mathrm{T}$, and the following matrices
\begin{align} \label{eq:gainK}
\Bar{A}_\mathrm{g} &= \Bar{A} + \Bar{B}K, \  K = \begin{bmatrix}
        \mathrm{diag}(K_{1i}) & \mathrm{diag}(K_{2i})
    \end{bmatrix} \nonumber \\  
\Bar{A} &= \begin{bmatrix}
        \mathrm{diag}(\Bar{A}_i) & 0 \\
        \mathrm{diag}(G_{2i})\mathcal{W}\mathrm{diag}(\Bar{C}_i) & \mathrm{diag}(G_{1i})
    \end{bmatrix} \nonumber 
    \\
      \Bar{B} &= \begin{bmatrix}
        \mathrm{diag}(\Bar{B}_{i}) \\ \mathrm{diag}(G_{2i})\mathcal{W}\mathrm{diag}(\Bar{D}_i)
    \end{bmatrix}, \ B_\mathrm{g} = \begin{bmatrix}
        \mathrm{diag}(E_{i}) \\ -\mathrm{diag}(G_{2i})\mathcal{W}F_\mathrm{a}
    \end{bmatrix} \nonumber \\ 
    \Bar{C}_{\mathrm{g}} &= \begin{bmatrix}
        \mathrm{diag}(\Bar{C}_{i}+{\Bar{D}_i}K_{1i}) & \mathrm{diag}({\Bar{D}_i}K_{2i})
    \end{bmatrix} \nonumber \\
       D_\mathrm{g} &= -F_\mathrm{a}, \ F_\mathrm{a} = I_N \otimes F,  \ A_{0\mathrm{a}} = I_N \otimes A_0  \nonumber \\    \mathcal{F} &= \mathrm{diag}((d_i+g_i)^{-1}), \ \mathcal{W} = (I_N-\mathcal{F}\mathcal{A})\otimes I_p. 
\end{align} 
The nominal closed-loop system and output matrices are denoted by ${A}_\mathrm{g}$ and ${C}_\mathrm{g}$, respectively. 
The matrix  ${A}_\mathrm{g}$ can be written as ${A}_\mathrm{g}  = A + B K$, where $A$ and $B$ denote the nominal (i.e., $\delta = 0$) parts of $\Bar{A}$  and $\Bar{B}$, respectively.

We now tailor Lemma 7.3 in \cite{cai2022cooperative} for the system \eqref{eq:composite_system}.

\begin{lemma}\label{lmm:fundamentallemma}
If $A_\mathrm{g}$ is Schur and for any $\mathrm{diag}(E_i)$, $F$, and for any $\delta \in \Delta$, where $\Delta = \{ \delta  ~|~   \Bar{A}_{\mathrm{g}} \text{ is} \ \text{Schur}\}$, 
    there exists a matrix 
    $X_\mathrm{g} $ such that
    \begin{align}\label{eq:clsdloopreg} \nonumber
            X_\mathrm{g}A_{0\mathrm{a}} &= \Bar{A}_\mathrm{g}X_\mathrm{g} + B_\mathrm{g} \\
            0 &= \Bar{C}_\mathrm{g}X_\mathrm{g} + D_\mathrm{g}  
        \end{align} 
then Problem \ref{prb:mainproblem} is solved. 
\end{lemma}

The linear matrix equations (LMEs) \eqref{eq:clsdloopreg}, which depend on the unknown matrices $\mathrm{diag}(E_i)$, $F$, and the uncertain parameter $\delta$, sets an algebraic goal of the distributed control law  \eqref{eq:controllaw}. To meet this goal, we use the following two lemmas. 

\begin{lemma}[Lemma 4.1 and Remark 4.1 in \cite{SarsilmazIJC}]\label{lmm:Wnonsingular}
    Under Condition \ref{ass:spanningtree}, $\mathcal{W}$ is nonsingular.
\end{lemma}

\begin{lemma}\label{lmm:weirdlemma}
    Let Condition \ref{ass:pcopy} hold. If there exists a pair $(Y_\mathrm{g},Z_\mathrm{g}) $ such that 
\begin{align}\label{eq:weirdequation}
        Z_\mathrm{g}A_{0\mathrm{a}} = \mathrm{diag}(G_{1i})Z_\mathrm{g} + \mathrm{diag}(G_{2i})Y_\mathrm{g}
    \end{align}
 then    $Y_\mathrm{g} = 0$. 
\end{lemma}
\begin{IEEEproof}
   Due to the block-diagonal structure of $A_{0\mathrm{a}}$, where each block on the diagonal is $A_0$, the minimal polynomials for $A_{0\mathrm{a}}$ and $A_0$ are the same. Thus, under Condition \ref{ass:pcopy}, the pair $(\mathrm{diag}(G_{1i}),\mathrm{diag}(G_{2i}))$ incorporates an $Np$-copy internal model of $A_{0\mathrm{a}}$. By Lemma 7.4 in \cite{cai2022cooperative}, the proof is over.
\end{IEEEproof}

Lemma 4.3 in \cite{SarsilmazIJC} has extended Lemma 1.27 of \cite{huang2004nonlinearbook}, which is key for $p$-copy internal model-based control of single systems, to internal model-based distributed control of continuous-time heterogeneous (in dimension) MASs over general directed graphs. Within the context of RCORP, we present the following key lemma, the discrete-time counterpart of the result in \cite{SarsilmazIJC} for the distributed control law \eqref{eq:controllaw}.

\begin{lemma}\label{lmm:keylemma}
    Let Conditions \ref{ass:spanningtree}--\ref{ass:pcopy} hold. 
    %
    For any $\mathrm{diag}(E_i)$, $F$, and for any $\delta \in \Delta$, 
    there exists a unique matrix $X_\mathrm{g}$ that satisfies the LMEs \eqref{eq:clsdloopreg}. 
\end{lemma}
\begin{IEEEproof}
 Fix $\mathrm{diag}(E_i)$, $F$, and let $\delta \in \Delta$. Then   $\Bar{A}_\mathrm{g}$ is Schur.   
  Due to the  structure of $A_{0\mathrm{a}}$, $\mathrm{spec}(A_{0\mathrm{a}}) = \mathrm{spec}(A_0)$. 
  Under Condition \ref{ass:A0antiSchur}, $A_{0\mathrm{a}}$ and $\Bar{A}_\mathrm{g}$ have no eigenvalues in common.  Thus, 
the Sylvester equation (i.e., the first LME) in \eqref{eq:clsdloopreg}  has a unique solution $X_\mathrm{g}$ by Proposition A.2 in \cite{huang2004nonlinearbook}.
In addition, we show that the matrix $X_\mathrm{g}$ satisfies the second LME in \eqref{eq:clsdloopreg}. To this end,  partition $X_\mathrm{g}$ as $X_\mathrm{g} = [X^\mathrm{T} \ Z^\mathrm{T}]^\mathrm{T}$, where $X \in \mathbb{R}^{\Bar{n}\times Nn_0}$, $Z \in \mathbb{R}^{Nn_z \times Nn_0}$, and $\Bar{n} = \sum_{i=1}^{N} n_i$. The Sylvester equation can now be expanded as  
\begin{align}\label{eq:sylvesterexpanded_1}
        XA_{0\mathrm{a}} &= \mathrm{diag}(\Bar{A}_i+\Bar{B}_iK_{1i})X + \mathrm{diag}(\Bar{B}_iK_{2i})Z + \mathrm{diag}(E_i) \nonumber \\
        ZA_{0\mathrm{a}} &= \mathrm{diag}(G_{1i})Z + \mathrm{diag}(G_{2i})\mathcal{W}Y
    \end{align}
where 
\begin{align}
         Y = \Bar{C}_\mathrm{g}X_\mathrm{g} + D_\mathrm{g}. \nonumber 
    \end{align}
Observe that the second equation in \eqref{eq:sylvesterexpanded_1} is in the form of \eqref{eq:weirdequation}. Under Condition \ref{ass:pcopy}, we infer from Lemma \ref{lmm:weirdlemma} that $\mathcal{W}Y = 0$. 
Under Condition \ref{ass:spanningtree}, by Lemma \ref{lmm:Wnonsingular}, $\mathcal{W}$ is nonsingular. Therefore, $\mathcal{W}Y = 0$ implies  $Y = 0$. This completes the proof. 
\end{IEEEproof}




The following theorem, which is the discrete-time version of 
Theorem 4.1 in \cite{SarsilmazIJC}, provides sufficient conditions for the solvability of the RCORP.






\begin{theorem}\label{thm:keytheorem}
  Let Conditions \ref{ass:spanningtree}--\ref{ass:pcopy} hold. If $A_\mathrm{g}$ is Schur, then Problem \ref{prb:mainproblem} is solved.
\end{theorem}
\begin{IEEEproof}
   The proof follows  from Lemmas \ref{lmm:fundamentallemma} and \ref{lmm:keylemma}. 
\end{IEEEproof}

\begin{remark}\label{rmk:sufficient_necessary}
The first equation in \eqref{eq:controllaw}, called the distributed $p$-copy internal model of $A_0$ under Condition \ref{ass:pcopy}  (e.g., see Chapter 9 in \cite{cai2022cooperative}), leads to the augmented pair $(A, B)$. Theorem \ref{thm:keytheorem}  tells us that under Conditions \ref{ass:spanningtree}--\ref{ass:pcopy}, to solve Problem \ref{prb:mainproblem}, it is sufficient to solve the following stabilization problem: \textit{Design the structured $K$ in \eqref{eq:gainK} such that $A_{\mathrm{g}} = A+BK$ is Schur.}
Although the stabilizability of the pair $(A,B)$ is necessary for such a $K$ to exist, as Example \ref{ex_structured_stabilizability} shows,  it is not sufficient due to the zero entries of $K$.

\end{remark}


Under mild conditions, the following lemma establishes the stabilizability of the pair $(A,B)$. Compared to the single system case (e.g., see Lemma 1.37 in \cite{huang2004nonlinearbook}), the graph connectivity condition, Condition \ref{ass:spanningtree}, is the only additional one. 


\begin{lemma}\label{lmm:augsysstab}
    Under Conditions \ref{ass:spanningtree} and \ref{ass:pcopy}--\ref{ass:AiBistab}, the pair $(A,B)$ is stabilizable.
\end{lemma}
\begin{IEEEproof}
Define $M(\lambda) =
[A - \lambda I_{\Bar{n}+Nn_z} \  B ]$, where $\Bar{n} = \sum_{i=1}^{N} n_i$. Let  $\lambda \in \mathbb{C}$ such that $|\lambda| \geq 1$.   According to the PBH test for stabilizability, it suffices to show   
   $\mathrm{rank} \hspace{0.05 cm} M(\lambda) = \Bar{n} + Nn_z$.
We first consider the case that $\lambda \notin  \mathrm{spec}(A_0)$. Under Condition \ref{ass:pcopy},  $\mathrm{spec}(\mathrm{diag}(G_{1i})) = \mathrm{spec}(A_0)$. Therefore, $\mathrm{det}(\mathrm{diag}(G_{1i})-\lambda I_{N n_z}) \neq  0$. Under Condition \ref{ass:AiBistab},  $(\mathrm{diag}(A_i),\mathrm{diag}(B_i))$ is stabilizable. Thus, by the PBH test,  $\mathrm{rank} \hspace{0.05 cm} [
    \mathrm{diag}(A_i) - \lambda I_{\bar n} \ \mathrm{diag}(B_i) ] = \bar n$. One can now conclude that $\mathrm{rank} \hspace{0.05 cm} M(\lambda) = \Bar{n} + Nn_z$. 

It remains to consider the case that $\lambda \in  \mathrm{spec}(A_0)$. We can write $M(\lambda) = M_1(\lambda) M_2(\lambda)$, where
    \begin{align}\nonumber
        M_1(\lambda) &= \begin{bmatrix}
            I_{\Bar{n}} & 0 & 0 \\ 0 & \mathrm{diag}(G_{2i})\mathcal{W} & \mathrm{diag}(G_{1i}) - \lambda I_{Nn_z}
        \end{bmatrix}
\nonumber \\
        M_2(\lambda) &= \begin{bmatrix}
            \mathrm{diag}(A_i)-\lambda I_{\Bar{n}} & 0 & \mathrm{diag}(B_i) \\ \mathrm{diag}(C_{i}) & 0 & \mathrm{diag}(D_i) \\
            0 & I_{Nn_z} & 0
        \end{bmatrix}. \nonumber
    \end{align}  
Under Condition \ref{ass:pcopy}, the pair $(\mathrm{diag}(G_{1i}),\mathrm{diag}(G_{2i}))$ is reachable. By the reachability matrix test,  $\mathrm{rank} \hspace{0.05 cm }\Gamma = Nn_z$, where
\begin{align}\nonumber
    \Gamma = \begin{bmatrix}
        \mathrm{diag}(G_{2i}) & \mathrm{diag}(G_{1i}G_{2i}) & \ldots & \mathrm{diag}(G_{1i}^{Nn_z -1}G_{2i})
    \end{bmatrix}.
\end{align}
Under Condition \ref{ass:spanningtree}, by Lemma \ref{lmm:Wnonsingular}, $\mathcal{W}$ is nonsingular. Then $I_{Nn_z} \otimes \mathcal{W}$ is nonsingular. In conjunction with $\mathrm{rank} \hspace{0.05 cm }\Gamma = Nn_z$, this implies  $\mathrm{rank} \hspace{0.05 cm }\Gamma (I_{Nn_z} \otimes \mathcal{W}) = Nn_z$. Observe that $\Gamma (I_{Nn_z} \otimes \mathcal{W})$ is the reachability matrix of the pair $(\mathrm{diag}(G_{1i}),\mathrm{diag}(G_{2i})\mathcal{W})$. By the reachability matrix test,  the pair $(\mathrm{diag}(G_{1i}),\mathrm{diag}(G_{2i})\mathcal{W})$ is reachable. From the PBH test, $\mathrm{rank} \hspace{0.05 cm} 
    [\mathrm{diag}(G_{1i}) - \lambda I_{Nn_z} \ \mathrm{diag}(G_{2i})\mathcal{W}]
= Nn_z$. It is now clear from the structure of $M_1(\lambda)$ that $\mathrm{rank} \hspace{0.05 cm} M_1(\lambda) = \Bar{n} + Nn_z$. 
Under Condition \ref{ass:transzero}, it is also clear from the structure of $M_2(\lambda)$ that 
$\mathrm{rank} \hspace{0.05 cm} M_2(\lambda) = \Bar{n} + Nn_z + Np$. Due to the size of $M(\lambda)$, $\mathrm{rank} \hspace{0.05  cm} M(\lambda) \leq  \Bar{n} + Nn_z$. But, by Sylvester's inequality,  $\mathrm{rank} \hspace{0.05  cm} M(\lambda) \geq \Bar{n} + Nn_z$
Hence, $\mathrm{rank} \hspace{0.05  cm} M(\lambda) = \Bar{n} + Nn_z$. 
\end{IEEEproof}


\begin{example}\label{ex_structured_stabilizability}
Let  $A_i = 0.5$, $B_i =0$, and $C_i = D_i = 1$ for $i = 1, 2$, and $A_0 = 10$. Consider the adjacency matrix $\mathcal{A}$ with  $ a_{21} = a_{12} =  1$, and the pinning gains $g_{1} = 1$ and $g_{2} = 0$.
Select $G_{1i} = G_{2i} = 10$ for $i = 1, 2$. Note that Conditions \ref{ass:spanningtree}--\ref{ass:AiBistab} hold.
 By Lemma \ref{lmm:augsysstab}, the pair $(A,B)$ is stabilizable. Through  structured $K$ in \eqref{eq:gainK}, the eigenvalues of $A+BK$  are $\lambda_1 = \lambda_2 = 0.5$, $\lambda_{3} = 5(s+q) + 10$, and $\lambda_{4} = 5(s-q) + 10$, where  $s = K_{21} + K_{22}$ and $q = \sqrt{K_{21}^2 + K_{22}^2}$. By Cauchy-Schwarz inequality, one can show that $|s|/\sqrt{2} \leq q $. Suppose for a contradiction that there exist $K_{21}$ and $K_{22}$ such that $A+BK$ is Schur. Then $|\lambda_3| = |5(s+q) + 10| < 1$ and $ |\lambda_4| = |5(s-q) + 10| < 1$. Equivalently,
\begin{align}\nonumber
-2.2 < s+q <-1.8, \ \ -2.2 < s-q <-1.8.    
\end{align}
These two inequalities imply $s \in (-2.2,-1.8)$. Since $q \geq 0$, they also imply
$q\in[0,0.2)$. By considering $|s|/\sqrt{2} \leq q $,  $ 1.8< |s|$,  and $q < 0.2$, we obtain $1.8/ \sqrt{2} < 0.2$, a contradiction. Therefore, there is no structured $K$  that renders $A+BK$ Schur even though the pair $(A,B)$ is stabilizable.




\end{example}

\section{Structured Global Design of $K$} \label{sec:global_control}


 


Under Conditions \ref{ass:spanningtree}--\ref{ass:pcopy}, solving Problem \ref{prb:mainproblem} reduces to fulfillment of its property (i), as highlighted in Remark \ref{rmk:sufficient_necessary}. 
It is well known that  $A_{\mathrm{g}}$ is Schur if, and only if, there exists a $P \succ 0$ satisfying the discrete-time Lyapunov inequality, also referred to as the Stein inequality, $A_\mathrm{g}^\mathrm{T}PA_\mathrm{g} - P \prec 0$.   Since $\mathrm{spec}(A_{\mathrm{g}}) = \mathrm{spec}(A_{\mathrm{g}}^{\mathrm{T}}) $, 
$A_{\mathrm{g}}$ is Schur if, and only if, there exists a $P \succ 0$ satisfying the dual of the discrete time Lyapunov inequality $A_\mathrm{g}PA_\mathrm{g}^{\mathrm{T}} - P \prec 0$. The dual inequality is also called the Lyapunov inequality. 
Given $K$, both inequality constraint functions are convex in $P$. However, none of them are jointly convex in $K$ and $P$. Due to the structure of $K$, unlike control gains without any structure, we cannot directly obtain  
an equivalent LMI through a simple change of variables. 
On the other hand, imposing a structure on $P$ turns the Lyapunov inequality into a structured Lyapunov inequality. This makes the change of variables applicable and yields a convex formulation for the synthesis of the structured control gain $K$. The idea and its consequences are formulated in Lemma \ref{lmm:globaldesign} and Theorem \ref{thm:globaltheorem}, which are the discrete-time counterparts of Lemma 3 and Theorem 1 in \cite{koru2020cooperative}.











\begin{lemma}\label{lmm:globaldesign}
There exist matrices $K$ in the form of \eqref{eq:gainK} and 
\begin{align}\label{eq:Lyap_stability_factor}
P &= \begin{bmatrix}
            \mathrm{diag}(P_{1i}) & \mathrm{diag}(P_{\mathrm{o}i}) \\ \mathrm{diag}(P_{\mathrm{o}i}^{\mathrm{T}}) & \mathrm{diag}(P_{2i})
        \end{bmatrix} \succ 0
\end{align} 
where $P_{1i} \in \mathbb{S}^{n_i}_{++}$, $P_{2i} \in \mathbb{S}^{n_z}_{++}$ necessarily, and $P_{\mathrm{o}i} \in \mathbb{R}^{n_i \times n_z}$ for all $i\in \mathcal{N}$, such that 
    \begin{align}\label{eq:lyapunov}
   A_\mathrm{g}PA_\mathrm{g}^\mathrm{T} - P \prec 0
    \end{align}
if, and only if, there exist matrices 
\begin{align}\label{eq:matricesYandP} \nonumber
Q &= \begin{bmatrix}
            \mathrm{diag}(Q_{1i}) & \mathrm{diag}(Q_{\mathrm{o}i}) \\ \mathrm{diag}(Q_{\mathrm{o}i}^{\mathrm{T}}) & \mathrm{diag}(Q_{2i})
        \end{bmatrix} \succ 0 \\
        Y &= \begin{bmatrix}
            \mathrm{diag}(Y_{1i}) & \mathrm{diag}(Y_{2i})
        \end{bmatrix} 
    \end{align}
    where $Q_{1i} \in \mathbb{S}^{n_i}_{++}$, $Q_{2i} \in \mathbb{S}^{n_z}_{++}$ necessarily,  $Q_{\mathrm{o}i} \in \mathbb{R}^{n_i \times n_z}$ 
    $Y_{1i} \in \mathbb{R}^{m_i \times n_i}$, and $Y_{2i} \in \mathbb{R}^{m_i \times n_z}$ for all $i \in \mathcal{N}$, such that  
\begin{align}\label{eq:schurcomplementmatrix}
        \begin{bmatrix}
            - Q & AQ + BY \\ QA^\mathrm{T}+Y^\mathrm{T}B^\mathrm{T}  & -Q
        \end{bmatrix} \prec 0. 
    \end{align}
\end{lemma}

\begin{IEEEproof}
For any $Q \succ 0$ and $Y$ of appropriate sizes, the inequality \eqref{eq:schurcomplementmatrix} is equivalent to 
\begin{align}\label{eq:nonconvex_global_form}
(A+BYQ^{-1})Q(A+BYQ^{-1})^{\mathrm{T}} -Q \prec 0
\end{align}
by the Schur complement of $-Q$ in the second diagonal block of the matrix in \eqref{eq:schurcomplementmatrix}. 
To show the ``if'' part, suppose there exist $Q  $ and $Y$ in the form of \eqref{eq:matricesYandP} such that the inequality \eqref{eq:schurcomplementmatrix} holds.  It can be verified by Theorem 2.1 in \cite{LU2002119} that $Q^{-1}$ is in the form of $Q$. Hence, $YQ^{-1}$ satisfies the structure of $K$ in \eqref{eq:gainK}. Since \eqref{eq:schurcomplementmatrix} is equivalent to  \eqref{eq:nonconvex_global_form}, the Lyapunov inequality \eqref{eq:lyapunov} holds with $P = Q$ and $K = YQ^{-1}$. To show the ``only if'' part, assume there exist $K$ in the form of \eqref{eq:gainK} and $P$ in the form of 
\eqref{eq:Lyap_stability_factor} such that the Lyapunov inequality \eqref{eq:lyapunov} holds. Noting  $KP$ satisfies the structure of $Y$ in \eqref{eq:matricesYandP}, the inequality \eqref{eq:nonconvex_global_form} holds with $Q = P$ and $Y =  KP$. The proof is thus completed by the equivalence of \eqref{eq:nonconvex_global_form} and  \eqref{eq:schurcomplementmatrix}. 
\end{IEEEproof}

\begin{remark}\label{remark_stabilizability}
To interpret this result in terms of stabilizability,
the stabilizability of the pair $(A,B)$ through structured $K$ in \eqref{eq:gainK} will be referred to as the \textit{structured stabilizability of the pair $(A,B)$}. The Lyapunov inequality \eqref{eq:lyapunov} with structured $P$ in the form of \eqref{eq:Lyap_stability_factor} will be referred to as the \textit{structured Lyapunov inequality}. 
Though the LMI in Lemma \ref{lmm:globaldesign} does not characterize the structured stabilizability of the pair $(A,B)$ due to the restriction of $P$ to the form in \eqref{eq:Lyap_stability_factor}, it characterizes the  structured stabilizability of the pair $(A,B)$ with respect to the structured Lyapunov inequality.



\end{remark}




\begin{remark}
    With the following permutation matrix 
\begin{align} \label{eq:permutationT}
        T = \begin{bmatrix}
            \mathrm{diag}\left(\begin{bmatrix}
                I_{n_i} \\ 0_{n_{z_i} \times n_i}
            \end{bmatrix}\right) & \mathrm{diag}\left(\begin{bmatrix}
                0_{n_i \times n_{z_i} } \\ I_{n_{z_i}}
            \end{bmatrix}\right)
        \end{bmatrix}
    \end{align}
    where $n_{z_i} = n_z$ for all $i \in \mathcal{N}$, the matrix  $Q$ in the form of \eqref{eq:matricesYandP} is permutation similar to the matrix
   \begin{align}
       \mathrm{diag}\left(\begin{bmatrix}
            Q_{1i} & Q_{\mathrm{o}i} \\ Q_{\mathrm{o}i}^\mathrm{T} &  Q_{2i}
        \end{bmatrix}\right).  \nonumber
    \end{align}
    Since $Q \succ 0 $, this implies that, for any  $i \in \mathcal{N}$, 
    \begin{align}
        \begin{bmatrix}
            Q_{1i} & Q_{\mathrm{o}i} \\ Q_{\mathrm{o}i}^\mathrm{T} &  Q_{2i}
        \end{bmatrix} \succ 0.  \nonumber
    \end{align}
    Therefore, the matrix inverse in \eqref{eq:individualgainsK} is  well-defined. 
\end{remark}

 \begin{theorem}\label{thm:globaltheorem}
 If there exist  $Q$ and $Y$ in the form of \eqref{eq:matricesYandP} such that \eqref{eq:schurcomplementmatrix} holds, and if 
the control gains $K_{1i}$ and $K_{2i}$ 
are recovered as 
     \begin{align}\label{eq:individualgainsK}
       \begin{bmatrix}
             K_{1i} & K_{2i}
         \end{bmatrix} = \begin{bmatrix}
             Y_{1i} & Y_{2i}
         \end{bmatrix}\begin{bmatrix}
            Q_{1i} & Q_{\mathrm{o}i} \\ Q_{\mathrm{o}i}^\mathrm{T} &  Q_{2i}
        \end{bmatrix}^{-1} 
     \end{align}
    for all $i \in \mathcal{N}$,
    then $A_\mathrm{g}$ is Schur. If, in addition, Conditions \ref{ass:spanningtree}--\ref{ass:pcopy} hold, then Problem \ref{prb:mainproblem} is solved.
 \end{theorem}
\begin{IEEEproof}
 By \eqref{eq:individualgainsK}, $Y_{1i} = K_{1i}Q_{1i} + K_{2i} Q_{\mathrm{o}i}^{\mathrm{T}}$ and $Y_{2i} = K_{1i}Q_{\mathrm{o}i} + K_{2i} Q_{2i}$ for all $i\in \mathcal{N}$. This implies that $Y = KQ$, 
  where $K$ satisfies the structure in \eqref{eq:gainK}.  From the proof of the ``if'' part of  Lemma \ref{lmm:globaldesign}, the Lyapunov inequality \eqref{eq:lyapunov} holds with $P = Q$ and $K$. This proves the first statement of the theorem. The second statement follows from Theorem \ref{thm:keytheorem}. 
\end{IEEEproof}

\begin{remark}
    Through the structured Lyapunov inequality, 
    the control gains $K_{1i}$ and $K_{2i}$ of each follower are recovered in \eqref{eq:individualgainsK} if  the LMI in Lemma \ref{lmm:globaldesign} is feasible. 
This convex semidefinite feasibility program can be reliably solved using fast solvers with guaranteed polynomial-time convergence \cite{lmibook, cvxbook, nesterov2018}. 
\end{remark}

\section{Agent-Wise Local Design of $K$}\label{sec:agent_wise_local_control}

With the proposed structured global design in Section \ref{sec:global_control}, the synthesis of the control gains $K_{1i}$ and $K_{2i}$ for all $i\in \mathcal{N}$ becomes a convex program, and hence, computationally tractable. However, the synthesis is centralized (i.e., an authority must know the nominal parameters of all followers), even though the implementation of the control law \eqref{eq:controllaw} remains distributed. 
To design each follower's  control gains $K_{1i}$ and $K_{2i}$ independently of the other followers, this section focuses on the individual nominal dynamics of each follower.  








Considering \eqref{eq:MAS}--\eqref{eq:controllaw} and letting $v(0) = 0$, we arrive at the exosystem-free nominal local dynamics of followers  given by
\begin{align} \nonumber
    \xi_i^{+} &= A_{\mathrm{f}i}\xi_i + B_{\mathrm{f}i}\mu_i \\
    e_i &= C_{\mathrm{f}i}\xi_i, \quad i = 1,\ldots,N \nonumber
\end{align}
with $\xi_i = [
    x_i^\mathrm{T} \ z_i^\mathrm{T}
]^\mathrm{T}$, $\mu_i =(d_i+g_i)^{-1}\sum_{j=1}^{N} a_{ij}e_j$, and the following matrices
\begin{align}\nonumber
A_{\mathrm{f}i} &= A_{\mathrm{o}i}+B_{\mathrm{o}i}K_i, \ K_i = \begin{bmatrix}
        K_{1i} & K_{2i}
    \end{bmatrix}  \nonumber \\
    A_{\mathrm{o}i} &= \begin{bmatrix}
        A_i & 0 \\ G_{2i}C_i & G_{1i}
    \end{bmatrix}, \ B_{\mathrm{o}i} = \begin{bmatrix}
        B_i \\ G_{2i}D_i
    \end{bmatrix}, \ B_{\mathrm{f}i} = \begin{bmatrix}
        0 \\ -G_{2i}
    \end{bmatrix} \nonumber  \\ 
    C_{\mathrm{f}i} &= C_{\mathrm{o}i}+D_{i}K_i, \ 
    C_{\mathrm{o}i} = \begin{bmatrix}
        C_i & 0 
    \end{bmatrix}. 
    \nonumber 
\end{align}

Since $|\mathcal{E}'|<N$ and $d_i + g_i >0 $ for all $i\in \mathcal{N}$, we have $\mathrm{rank}(\mathcal{F}\mathcal{A}) > 0$. Hence,  
$\mathcal{F}\mathcal{A}$ has at least one nonzero singular value. 
 Let $\sigma_{\mathrm{min}}(\mathcal{F}\mathcal{A})$ and $\sigma_{\mathrm{max}}(\mathcal{F}\mathcal{A})$ denote the minimum \textit{nonzero} singular value  and maximum singular value of $\mathcal{F}\mathcal{A}$.
The following theorem now derives agent-wise local sufficient conditions for the matrix $A_\mathrm{g}$ to be Schur. 


\begin{theorem}\label{thm:agentwiseLocal_Cyclic} If, for any $i \in \mathcal{N}$, there exist $K_i$, $P_i \succ 0$, and
$r_i \geq \sigma_{\mathrm{max}}^3(\mathcal{F}\mathcal{A})/\sigma_\mathrm{min}(\mathcal{F}\mathcal{A}) $ such that 
\begin{align}\label{eq:agentwise_local_inequalities}
 &  \begin{bmatrix} 
A_{\mathrm{f}i}P_iA_{\mathrm{f}i}^\mathrm{T} \hspace{-0.02cm}-\hspace{-0.02cm} P_i \hspace{-0.02cm}+\hspace{-0.02cm} 
r_iB_{\mathrm{f}i}(I_p \hspace{-0.02cm}+\hspace{-0.02cm}  C_{\mathrm{f}i}P_iC_{\mathrm{f}i}^\mathrm{T})B_{\mathrm{f}i}^\mathrm{T}
& A_{\mathrm{f}i}P_iC_{\mathrm{f}i}^\mathrm{T} \\ C_{\mathrm{f}i}P_iA_{\mathrm{f}i}^\mathrm{T} & -I_p
    \end{bmatrix} 
    \prec 0  \nonumber \\
 & \hspace{0.18 cm}  \sigma_\mathrm{min}(\mathcal{F}\mathcal{A})I_p       \preceq C_{\mathrm{f}i}P_iC_{\mathrm{f}i}^\mathrm{T} \preceq \sigma_\mathrm{max}(\mathcal{F}\mathcal{A})I_p 
\end{align}
then  $A_\mathrm{g}$ is Schur. If, in addition, Conditions \ref{ass:spanningtree}--\ref{ass:pcopy} hold, then Problem \ref{prb:mainproblem} is solved.
\end{theorem}
\begin{IEEEproof}
For any $K_i$ and  symmetric $P_i$ of appropriate sizes, and for any $r_i \in \mathbb{R}$,  the inequalities \eqref{eq:agentwise_local_inequalities} are equivalent to 
\begin{align}\label{eq:ARElike}
M_i  \triangleq \hspace{0.1 cm}& A_{\mathrm{f}i}P_iA_{\mathrm{f}i}^\mathrm{T} - P_i + A_{\mathrm{f}i}P_iC_{\mathrm{f}i}^\mathrm{T}C_{\mathrm{f}i}P_iA_{\mathrm{f}i}^\mathrm{T} + r_iB_{\mathrm{f}i}B_{\mathrm{f}i}^\mathrm{T} \nonumber \\ &+ r_iB_{\mathrm{f}i}C_{\mathrm{f}i}P_iC_{\mathrm{f}i}^\mathrm{T}B_{\mathrm{f}i}^\mathrm{T} \prec 0 \nonumber  \\
& \sigma_\mathrm{min}(\mathcal{F}\mathcal{A}) \leq  \lambda_{\mathrm{min}}(C_{\mathrm{f}i}P_iC_{\mathrm{f}i}^\mathrm{T}) \nonumber  \\
& \lambda_{\mathrm{max}}(C_{\mathrm{f}i}P_iC_{\mathrm{f}i}^\mathrm{T}) \leq \sigma_\mathrm{max}(\mathcal{F}\mathcal{A})
\end{align}
by the Schur complement of $-I_p$ in the matrix of the first inequality in \eqref{eq:agentwise_local_inequalities} and the symmetry of $C_{\mathrm{f}i}P_iC_{\mathrm{f}i}^\mathrm{T}$. 
Thus, to prove the lemma, for any $i\in \mathcal{N}$,  let $K_i$, $P_i\succ 0$, and $r_i \geq \sigma_{\mathrm{max}}^3(\mathcal{F}\mathcal{A})/\sigma_\mathrm{min}(\mathcal{F}\mathcal{A})$ be such that the inequalities \eqref{eq:ARElike} hold. 

For any $i \in \mathcal{N}$, partition $P_i$ as 
\begin{align}\label{eq:small_P_partition}
        P_{i} = \begin{bmatrix}
            P_{1i} & P_{\mathrm{o}i} \\  P_{\mathrm{o}i}^\mathrm{T} &  P_{2i}
        \end{bmatrix} 
    \end{align}
where $\ P_{1i} \in \mathbb{S}^{n_i}_{++}$, $P_{2i} \in \mathbb{S}^{n_z}_{++}$ necessarily, and $P_{\mathrm{o}i} \in \mathbb{R}^{n_i \times n_z}$. 
With the permutation matrix  \eqref{eq:permutationT}, define $P = T^\mathrm{T}\mathrm{diag}(P_{i})T $. Note that $T^{\mathrm{T}} = T^{-1}$.
Therefore, $P$ is permutation similar to $\mathrm{diag}(P_{i})$. Since $P_i \succ 0$ for all $i \in \mathcal{N}$, this ensures that $P$ is in the form of \eqref{eq:Lyap_stability_factor}. To show that the Lyapunov inequality \eqref{eq:lyapunov} holds, we similarly
 define $M = T^\mathrm{T}\mathrm{diag}(M_{i})T$.
In conjunction with the first two lines of \eqref{eq:ARElike}, this yields
\begin{align}\label{eq:M_compact}
           M = \hspace{0.1 cm}&  \hat{A}P\hat{A}^\mathrm{T} - P + \hat{A}PC_\mathrm{g}^\mathrm{T}C_\mathrm{g}P\hat{A}^\mathrm{T} + \hat{B}(R \otimes I_p)\hat{B}^\mathrm{T}  \nonumber \\ &+ \hat{B}(R^{1/2} \otimes I_p)C_\mathrm{g}PC_\mathrm{g}^\mathrm{T}(R^{1/2} \otimes I_p)\hat{B}^\mathrm{T} \prec 0 \end{align}
where 
\begin{align}
        \hat{A} &= \begin{bmatrix}
            \mathrm{diag}(A_i + B_iK_{1i}) & \mathrm{diag}(B_iK_{2i}) \\ \mathrm{diag}(G_{2i}(C_{i}+D_iK_{1i})) & \mathrm{diag}(G_{1i} + G_{2i}D_iK_{2i})
        \end{bmatrix} \nonumber \\
        \hat{B} &= \begin{bmatrix}
            0 \\ -\mathrm{diag}(G_{2i})
        \end{bmatrix}, \ R = \mathrm{diag}(r_i). \nonumber
    \end{align}
    
 Note that $C_{\mathrm{g}} =  \mathrm{diag}(C_{\mathrm{f}i})T $. It can be easily verified that $C_{\mathrm{g}} P C_{\mathrm{g}}^{\mathrm{T}}  = \mathrm{diag}(C_{\mathrm{f}i}P_i C_{\mathrm{f}i}^{\mathrm{T}})$. Hence,  
   \begin{align}\nonumber
        \mathrm{spec}(C_\mathrm{g}PC_\mathrm{g}^\mathrm{T}) = \bigcup_{i=1}^{N} \mathrm{spec}(C_{\mathrm{f}i}P_iC_{\mathrm{f}i}^\mathrm{T}).
    \end{align}
    Since  $\sigma_\mathrm{min}(\mathcal{F}\mathcal{A}) \leq \lambda_{\mathrm{min}}(C_{\mathrm{f}i}P_iC_{\mathrm{f}i}^\mathrm{T})$  and $\lambda_{\mathrm{max}}(C_{\mathrm{f}i}P_iC_{\mathrm{f}i}^\mathrm{T}) \leq \sigma_\mathrm{max}(\mathcal{F}\mathcal{A})$ for all $i \in \mathcal{N}$, this, together with the symmetry of $C_{\mathrm{g}} P C_{\mathrm{g}}^{\mathrm{T}} $, implies
    \begin{align}
    \sigma_\mathrm{min}(\mathcal{F}\mathcal{A}) \leq \lambda_{\mathrm{min}}(C_{\mathrm{g}}PC_{\mathrm{g}}^\mathrm{T}) \leq 
        \lambda_{\mathrm{max}}(C_{\mathrm{g}}PC_{\mathrm{g}}^\mathrm{T}) \leq \sigma_\mathrm{max}(\mathcal{F}\mathcal{A}). \nonumber   
    \end{align}
From these inequalities, we have 
\begin{align}
    \frac{\lambda_{\mathrm{max}}(C_{\mathrm{g}}PC_{\mathrm{g}}^\mathrm{T})}{\lambda_{\mathrm{min}}(C_{\mathrm{g}}PC_{\mathrm{g}}^\mathrm{T})} \leq \frac{\sigma_\mathrm{max}(\mathcal{F}\mathcal{A})}{\sigma_\mathrm{min}(\mathcal{F}\mathcal{A})} \nonumber
\end{align}
which yields 
\begin{align}\label{eq:inequality_FA_CgPCgT}
\sigma_\mathrm{max}^2(\mathcal{F}\mathcal{A})  \frac{\lambda_{\mathrm{max}}(C_{\mathrm{g}}PC_{\mathrm{g}}^\mathrm{T})}{\lambda_{\mathrm{min}}(C_{\mathrm{g}}PC_{\mathrm{g}}^\mathrm{T})} \leq r_i, \quad i= 1,\ldots, N
\end{align}
because  $r_i \geq \sigma_{\mathrm{max}}^3(\mathcal{F}\mathcal{A})/\sigma_\mathrm{min}(\mathcal{F}\mathcal{A})$ for all $i\in \mathcal{N}$. 
Due to the symmetry of $C_{\mathrm{g}}P C_{\mathrm{g}}^{\mathrm{T}} $, we have 
\begin{align}\label{eq:direct_consequence_of_symmetry}
    \lambda_{\mathrm{min}}(C_\mathrm{g}PC_\mathrm{g}^\mathrm{T})I_{Np} \preceq  C_\mathrm{g}P C_\mathrm{g}^\mathrm{T}  \preceq \lambda_{\mathrm{max}}(C_\mathrm{g}PC_\mathrm{g}^\mathrm{T}) I_{Np}. 
\end{align}
By the first inequality in \eqref{eq:direct_consequence_of_symmetry}, we obtain  
\begin{align}
   &   \lambda_{\mathrm{min}}(C_\mathrm{g}PC_\mathrm{g}^\mathrm{T}) (R \otimes I_p) \preceq  (R^{1/2} \otimes I_p) C_\mathrm{g}P C_\mathrm{g}^\mathrm{T} (R^{1/2} \otimes I_p). \nonumber 
\end{align}
 Through \eqref{eq:inequality_FA_CgPCgT}, this yields 
\begin{align}\label{eq:singular_eigenvalue_CgPCgT}
 & \sigma_\mathrm{max}^2(\mathcal{F}\mathcal{A})\lambda_{\mathrm{max}}(C_{\mathrm{g}}PC_{\mathrm{g}}^\mathrm{T})I_{Np} \preceq \nonumber \\ 
&   (R^{1/2} \otimes I_p) C_\mathrm{g}P C_\mathrm{g}^\mathrm{T} (R^{1/2} \otimes I_p).
\end{align}
 From the second inequality in \eqref{eq:direct_consequence_of_symmetry}, we have 
\begin{align}\label{eq:FA_CgPCgT_max_eigenv}
 &   (\mathcal{F}\mathcal{A} \otimes I_p)C_\mathrm{g}P C_\mathrm{g}^\mathrm{T} (\mathcal{F}\mathcal{A} \otimes I_p)^{\mathrm{T}}  \preceq  \nonumber \\ & \lambda_{\mathrm{max}}(C_\mathrm{g}PC_\mathrm{g}^\mathrm{T})(\mathcal{F}\mathcal{A} \mathcal{A}^\mathrm{T} \mathcal{F}^{\mathrm{T}} \otimes I_p). 
\end{align}
Moreover, we have the following set equalities
\begin{align}\label{eq:spec_singular_values}
     \mathrm{spec} (\mathcal{F}\mathcal{A} \mathcal{A}^{\mathrm{T}}\mathcal{F}^{\mathrm{T}} \otimes I_p)  &=   \mathrm{spec} (\mathcal{F}\mathcal{A} \mathcal{A}^\mathrm{T} \mathcal{F}^{\mathrm{T}})  \nonumber \\  & =  \mathrm{spec} (\mathcal{A}^\mathrm{T} \mathcal{F}^{\mathrm{T}}\mathcal{F}\mathcal{A}) 
\end{align}
where the first equality follows from Proposition 7.1.10 in \cite{bernstein11} and the second  holds since $\mathcal{F}\mathcal{A}$ is a square matrix. 
In conjunction with the symmetry of $\mathcal{F}\mathcal{A} \mathcal{A}^{\mathrm{T}}\mathcal{F}^{\mathrm{T}} \otimes I_p$, 
\eqref{eq:spec_singular_values} implies 
\begin{align}\label{eq:FA_max_sing_square}
    \mathcal{F}\mathcal{A} \mathcal{A}^\mathrm{T} \mathcal{F}^{\mathrm{T}} \otimes I_p \preceq  \sigma_\mathrm{max}^2(\mathcal{F}\mathcal{A}) I_{Np}. 
\end{align}

We now infer from \eqref{eq:singular_eigenvalue_CgPCgT}, \eqref{eq:FA_CgPCgT_max_eigenv}, and \eqref{eq:FA_max_sing_square} that 
  \begin{align}\label{eq:FA_CgPCgT}
      &  (\mathcal{F}\mathcal{A} \otimes I_p)C_\mathrm{g}PC_\mathrm{g}^\mathrm{T}(\mathcal{F}\mathcal{A} \otimes I_p)^\mathrm{T} \preceq  \nonumber \\ 
      &  (R^{1/2} \otimes I_p)C_\mathrm{g}PC_\mathrm{g}^\mathrm{T}(R^{1/2} \otimes I_p).   
    \end{align}
    Since $r_i \geq \sigma_{\mathrm{max}}^3(\mathcal{F}\mathcal{A})/\sigma_\mathrm{min}(\mathcal{F}\mathcal{A})$ for all $i\in \mathcal{N}$, one can derive from \eqref{eq:FA_max_sing_square}   that  
    \begin{align}\label{eq:FA_R}
    \mathcal{F}\mathcal{A} \mathcal{A}^\mathrm{T} \mathcal{F}^{\mathrm{T}} \otimes I_p \preceq  R \otimes I_p.
\end{align}
By Lemma 6.2 in \cite{Gu2003}, 
for any real matrices $U$ and $V$ of appropriate sizes, $UV + V^\mathrm{T}U^{\mathrm{T}} \preceq UU^\mathrm{T} + V^\mathrm{T}V$. 
Letting $U = \hat B (\mathcal{F}\mathcal{A} \otimes I_p)$ and $V = C_\mathrm{g}P\hat{A}^\mathrm{T} $ gives 
\begin{align}\label{eq:useful_inequality}
    & \hat B (\mathcal{F}\mathcal{A} \otimes I_p)  C_\mathrm{g}P\hat{A}^\mathrm{T} +  \hat{A} P C_\mathrm{g}^{\mathrm{T}} (\mathcal{F}\mathcal{A} \otimes I_p)^{\mathrm{T}}  \hat B^{\mathrm{T}}    \preceq  \nonumber  \\
    &  \hat B (\mathcal{F}\mathcal{A} \mathcal{A}^\mathrm{T} \mathcal{F}^{\mathrm{T}} \otimes I_p) \hat B^{\mathrm{T}} + \hat{A} P C_\mathrm{g}^{\mathrm{T}} C_\mathrm{g}P\hat{A}^\mathrm{T}.
\end{align}
Utilizing \eqref{eq:FA_CgPCgT}--\eqref{eq:useful_inequality}, we obtain from \eqref{eq:M_compact} that 
\begin{align}
 &   \hat{A}P\hat{A}^\mathrm{T}  +  \hat B (\mathcal{F}\mathcal{A} \otimes I_p)  C_\mathrm{g}P\hat{A}^\mathrm{T} + \hat{A} P C_\mathrm{g}^{\mathrm{T}} (\mathcal{F}\mathcal{A} \otimes I_p)^{\mathrm{T}}  \hat B^{\mathrm{T}} \nonumber \\
 &+ \hat B(\mathcal{F}\mathcal{A} \otimes I_p)C_\mathrm{g}PC_\mathrm{g}^\mathrm{T}(\mathcal{F}\mathcal{A} \otimes I_p)^\mathrm{T} \hat B^{\mathrm{T}} - P \prec 0  \nonumber 
\end{align}
which can be equivalently written as  
\begin{align}
(\hat{A}+\hat{B}(\mathcal{F}\mathcal{A} \otimes I_p)C_\mathrm{g})P(\hat{A}+\hat{B}(\mathcal{F}\mathcal{A} \otimes I_p)C_\mathrm{g})^{\mathrm{T}} -P \prec 0. \nonumber 
\end{align}
Noting  $A_\mathrm{g} = \hat{A} + \hat{B}(\mathcal{F}\mathcal{A} \otimes I_p)C_\mathrm{g}$, this establishes that $A_{\mathrm{g}}$ is Schur. The second statement follows from Theorem \ref{thm:keytheorem}. 
\end{IEEEproof}

\begin{remark}\label{remark_convex_nonconvex_local}
Given $r_i \geq  \sigma_{\mathrm{max}}^3(\mathcal{F}\mathcal{A})/\sigma_\mathrm{min}(\mathcal{F}\mathcal{A})  $ and $K_i$, the inequalities \eqref{eq:agentwise_local_inequalities} are LMIs in $P_i \succ 0$. 
However, the feasible set becomes nonconvex when $K_i$ is treated as a decision variable.







\end{remark}


The following corollary presents a convexification that yields a convex semidefinite feasibility program for each follower to synthesize $K_i$ when $D_i = 0$.   

\begin{corollary}\label{crl:ControlSynthesisDiszero}
Let $D_i =0 $ and $r_i \geq  \sigma_{\mathrm{max}}^3(\mathcal{F}\mathcal{A})/\sigma_\mathrm{min}(\mathcal{F}\mathcal{A})  $ for all $i\in \mathcal{N}$. If, for any $i \in \mathcal{N}$, there exist  $P_i \succ 0$, $Y_i$, and  $\Theta_i$ such that 
 \begin{align}\label{eq:convex_general_D_LMI}
       &\begin{bmatrix}
            \Theta_i & Y_i \\ Y_i^{\mathrm{T}} & P_i
        \end{bmatrix} \succeq 0 \nonumber  \\
      &  \begin{bmatrix}
            \Omega_i  & (A_{\mathrm{o}i}P_i + B_{\mathrm{o}i}Y_i)C_{\mathrm{o}i}^{\mathrm{T}}  \\   C_{\mathrm{o}i}(P_iA_{\mathrm{o}i}^{\mathrm{T}} + Y_i^{\mathrm{T}}B_{\mathrm{o}i}^{\mathrm{T}}) &  -I_{p}
        \end{bmatrix} \prec 0  \nonumber \\
        & \hspace{0.18 cm} \sigma_\mathrm{min}(\mathcal{F}\mathcal{A})I_p \preceq  C_{\mathrm{o}i}P_i C_{\mathrm{o}i}^{\mathrm{T}} \preceq \sigma_\mathrm{max}(\mathcal{F}\mathcal{A})I_p      
    \end{align}
where $\Theta_i \succeq 0$ necessarily and 
    \begin{align}
        \Omega_i = \hspace{0.1 cm}& A_{\mathrm{o}i}P_iA_{\mathrm{o}i}^{\mathrm{T}}  +B_{\mathrm{o}i}Y_iA_{\mathrm{o}i}^{\mathrm{T}} + A_{\mathrm{o}i}Y_i^{\mathrm{T}}B_{\mathrm{o}i}^{\mathrm{T}}+B_{\mathrm{o}i}\Theta_iB_{\mathrm{o}i}^{\mathrm{T}} - P_i \nonumber \\
        &+ r_iB_{\mathrm{f}i}B_{\mathrm{f}i}^\mathrm{T}     
    +r_i B_{\mathrm{f}i}C_{\mathrm{o}i}P_iC_{\mathrm{o}i}^{\mathrm{T}}B_{\mathrm{f}i}^{\mathrm{T}} \nonumber 
    \end{align}
and if  the control gain $K_i$ is recovered as $K_i = Y_iP_i^{-1}$ for all $i \in \mathcal{N}$, then $A_\mathrm{g}$ is Schur. If, in addition, Conditions \ref{ass:spanningtree}--\ref{ass:pcopy} hold, then  Problem \ref{prb:mainproblem} is solved.
\end{corollary}
\begin{IEEEproof}
     Fix $i\in \mathcal{N}$. Clearly, it suffices to show that the first inequality in \eqref{eq:ARElike} holds. 
     Since $K_i = Y_iP_i^{-1}$, $Y_i = K_i P_i$. 
     Thus, the Schur complement of $P_i$ in the matrix of the first inequality in \eqref{eq:convex_general_D_LMI} gives $K_i P_i K_i^{\mathrm{T}} \preceq \Theta_i $. Since $D_i = 0$, $C_{\mathrm{f}i} =  C_{\mathrm{o}i} $. Using $C_{\mathrm{f}i} =  C_{\mathrm{o}i} $, $K_i P_i K_i^{\mathrm{T}} \preceq \Theta_i $, and  $Y_i = K_i P_i$, we arrive at the following inequality
     \begin{align}\label{eq:Omega}
        A_{\mathrm{f}i}P_iA_{\mathrm{f}i}^\mathrm{T} - P_i  + r_iB_{\mathrm{f}i}B_{\mathrm{f}i}^\mathrm{T} + r_iB_{\mathrm{f}i}C_{\mathrm{f}i}P_iC_{\mathrm{f}i}^\mathrm{T}B_{\mathrm{f}i}^\mathrm{T} \preceq \Omega_i. 
    \end{align}
Using $C_{\mathrm{f}i} =  C_{\mathrm{o}i} $ and $Y_i = K_i P_i$, from the second inequality in \eqref{eq:convex_general_D_LMI}, we have 
    \begin{align}%
A_{\mathrm{f}i}P_iC_{\mathrm{f}i}^\mathrm{T}C_{\mathrm{f}i}P_iA_{\mathrm{f}i}^\mathrm{T}  + \Omega_i \prec 0. \nonumber
\end{align}
This and \eqref{eq:Omega} now yield the first inequality in \eqref{eq:ARElike}. 
\end{IEEEproof}

As in the continuous-time case (see \cite{NoCyclePaper} and Remark 4.5 in  \cite{SarsilmazIJC}), the next corollary tells that the design of $K_i$ for each follower simplifies to making the nominal local system matrix $A_{\mathrm{f}i}$ Schur if the directed graph $\mathcal{G}$ is acyclic. The existence of such a $K_i$ is ensured under Conditions \ref{ass:pcopy}--\ref{ass:AiBistab} (e.g., see Lemma 1.37 in \cite{huang2004nonlinearbook}), and linear control design techniques, such as eigenvalue assignment and LQR, can be used to find one.


\begin{corollary}\label{crl:acyclic}
 Let the directed graph $\mathcal{G}$  be acyclic.
 Then the following statements are true: 
  \begin{enumerate}
        \item [(i)] The matrix $A_{\mathrm{g}}$ is Schur if, and only if, $A_{\mathrm{f}i}$ is Schur for all $i \in \mathcal{N}$.
        \item [(ii)] Let Conditions \ref{ass:A0antiSchur} and \ref{ass:pcopy} hold. If $A_{\mathrm{f}i}$ is Schur for all $i \in \mathcal{N}$, then Problem \ref{prb:mainproblem} is solved. 
    \end{enumerate}
\end{corollary}
\begin{IEEEproof}
The second statement follows from the first footnote,  the first statement of the corollary, and Theorem \ref{thm:keytheorem}. Therefore, it suffices to prove the first statement. We recall from the proof of Theorem \ref{thm:agentwiseLocal_Cyclic} that $A_\mathrm{g} = \hat{A} + \hat{B}(\mathcal{F}\mathcal{A} \otimes I_p)C_\mathrm{g}$. Let $A_{\mathrm{c}} = \mathrm{diag}(A_{\mathrm{f}i})+\mathrm{diag}(B_{\mathrm{f}i})(\mathcal{F}\mathcal{A} \otimes I_p)\mathrm{diag}(C_{\mathrm{f}i})$.
With the permutation matrix  \eqref{eq:permutationT}, we have $\mathrm{diag}(A_{\mathrm{f}i}) = T\hat{A}T^\mathrm{T}$, $\mathrm{diag}(B_{\mathrm{f}i}) = T\hat{B}$, and $\mathrm{diag}(C_{\mathrm{f}i}) = C_\mathrm{g}T^\mathrm{T}$. 
Thus, $TA_\mathrm{g}T^\mathrm{T} = A_{\mathrm{c}}$, which implies that $A_\mathrm{g}$  is similar to $A_{\mathrm{c}}$. Since the directed graph $\mathcal{G}$ is acyclic, one can relabel the nodes in $\mathcal{N}$, and hence the followers, such that $i> j$ if $(j,i) \in \mathcal{E}$. Then, without loss of generality, we assume that the adjacency matrix $\mathcal{A}$ is lower triangular. Since the diagonal entries of $\mathcal{A}$ are zero and $\mathcal{F}$ is a diagonal matrix, $\mathcal{FA}$ is a lower triangular matrix with zero diagonal entries. Consequently, $A_\mathrm{c}$ is a  block lower triangular matrix and its diagonal blocks are $A_{\mathrm{f}1}, \ldots, A_{\mathrm{f}N}$. In conjunction with $A_\mathrm{g}$  being similar to  $A_{\mathrm{c}}$, this implies that $\mathrm{spec}(A_\mathrm{g}) = \mathrm{spec}(\mathrm{diag}(A_{\mathrm{f}i}))$. Hence, the proof is over. 
\end{IEEEproof}

\section{Set Inclusions Among Control Gain Sets Emerging from Global and Local Perspectives}\label{sec_set_inclusions_global_local}
The results in Sections \ref{sec:global_control} and \ref{sec:agent_wise_local_control} yield several sets of control gains that make $A_{\mathrm{g}}$ Schur from both structured global and agent-wise local perspectives. Including the universal set denoted by $K_{\mathrm{G}}$, we investigate the relationships between these sets. To this end, the control gain sets of interest  are given by
\begin{align} \nonumber
K_\mathrm{G} &= \{K ~|~ \eqref{eq:lyapunov}   \text{ is feasible for some }  P \succ 0 \} \\
K_\mathrm{S} &= \{K  ~|~  \eqref{eq:lyapunov}   \text{ is feasible for some }  P   \text{ subject to }  \eqref{eq:Lyap_stability_factor}\} \nonumber \\
K_\mathrm{LA} &= \{K  ~|~  A_{\mathrm{f}i}   \text{ is Schur }  \text{for all } i \in \mathcal{N}\} \nonumber \\
K_{\mathrm{LC}}
&=
\left\{
\begin{array}{c l}\hspace{-0.25cm}
K\hspace{0.05cm}\!\left|\hspace{-0.45cm}\vphantom{
\begin{array}{l}
\eqref{eq:agentwise_local_inequalities}
\text{ is feasible for some }
r_i \ge
\dfrac{\sigma_{\mathrm{max}}^3(\mathcal{F}\mathcal{A})}
      {\sigma_{\mathrm{min}}(\mathcal{F}\mathcal{A})}
\\
\text{and } P_i \succ 0,
\ \forall i \in \mathcal{N}
\end{array}}
\right. &
\begin{array}{l}
\eqref{eq:agentwise_local_inequalities}
\text{ is feasible for some }
r_i \ge
\dfrac{\sigma_{\mathrm{max}}^3(\mathcal{F}\mathcal{A})}
      {\sigma_{\mathrm{min}}(\mathcal{F}\mathcal{A})}
\\
\text{and } P_i \succ 0,
\ \text{for all } i \in \mathcal{N}
\end{array}
\end{array}
\hspace{-0.35cm}\right\}. \nonumber
\end{align}

\begin{corollary}\label{crl:setinclusion1} The set inclusions $K_\mathrm{LC} \subseteq K_\mathrm{S} \subseteq K_\mathrm{G}$ hold.
\end{corollary}
\begin{IEEEproof}
     The inclusion $K_\mathrm{LC} \subseteq K_\mathrm{S}$ follows from the proof of Theorem \ref{thm:agentwiseLocal_Cyclic} and the inclusion $K_\mathrm{S} \subseteq K_\mathrm{G}$ is by definition. 
\end{IEEEproof}

One may also question whether the following inclusions hold:  (i) $K_\mathrm{G} \subseteq K_\mathrm{S}$,  (ii) $K_\mathrm{S} \subseteq K_\mathrm{LC}$, (iii) $K_\mathrm{LA} \subseteq K_\mathrm{G}$, and (iv) $K_\mathrm{G} \subseteq K_\mathrm{LA}$. As shown in Examples \ref{ex:KgdoesnotimplyKs}--\ref{ex:KGdoesnotimplyKa}, the answer is negative for each case.

\begin{example}\label{ex:KgdoesnotimplyKs}
Let  $A_i = 0$ and  $B_i = C_i = D_i = 1$ for $i = 1, 2$, and $A_0 = 2$.  Consider the same adjacency matrix and the pinning gains given in Example \ref{ex_structured_stabilizability}. Select $G_{1i} = 2$ and $G_{2i} = 1$ for $i = 1, 2$. With $K_{11} = 1$, $K_{12} = -0.9$, $K_{21} = -1$, and $K_{22}=-2$, $A_\mathrm{g}$ is Schur, and hence, $K \in K_\mathrm{G}$. Suppose for a contradiction that $K \in K_\mathrm{S}$. Then there exists a structured $P$ in the form of \eqref{eq:Lyap_stability_factor} such that \eqref{eq:lyapunov} holds. Thus,  $P_{1i}>0$ and $P_{2i}>0$ for $i=1,2$,  and 
\begin{align}\label{eq:ineqsinexample}\nonumber
P_{\mathrm{o}1} &> 0.5 P_{21} \nonumber \\
        P_{22} &< 0.0475P_{12}-0.9P_{\mathrm{o}2} \nonumber \\
         P_{22} &< 0.1P_{\mathrm{o}2}-0.0025P_{12}-4P_{11}-4P_{\mathrm{o}1} \nonumber  \\
         P_{22} &> P_{21} + 0.01P_{12}+4P_{11}-4P_{\mathrm{o}1}
     \end{align}
where the inequalities \eqref{eq:ineqsinexample} follow from the diagonal entries of $A_\mathrm{g}PA_\mathrm{g}^\mathrm{T} - P$ being negative. Since $P_{21} >0$, the first inequality in \eqref{eq:ineqsinexample} implies $P_{\mathrm{o}1} >0$. Since $P_{1i} >0$ for $i=1, 2$, $P_{22}>0$, and $P_{\mathrm{o}1} >0$, we infer from the third inequality in  \eqref{eq:ineqsinexample} that $P_{\mathrm{o}2} >0$.
The second inequality in \eqref{eq:ineqsinexample} is equivalent to
\begin{align}\label{eq:second_ineq_Ex2}
    0.0125P_{12}>P_{22}/3.8+0.9P_{\mathrm{o}2}/3.8. 
\end{align}
From the third and fourth inequalities in \eqref{eq:ineqsinexample}, we obtain
\begin{align}\label{eq:third_fourth_inequality_Ex2_implies}
    0.0125P_{12} < 0.1 P_{\mathrm{o}2} - 8P_{11} - P_{21}. 
\end{align}
We now deduce from \eqref{eq:second_ineq_Ex2} and \eqref{eq:third_fourth_inequality_Ex2_implies} that  $-0.52P_{\mathrm{o}2}>P_{22}+30.4P_{11}+3.8P_{21}$. 
In conjunction with $P_{11}> 0$, $P_{2i}>0$ for $i=1,2$, and $P_{\mathrm{o}2} >0$, this implies that $0 > 0$, a contradiction. Hence, $K_{\mathrm{G}} \nsubseteq K_{\mathrm{S}}$.
\end{example}

\begin{example}\label{ex:KsdoesnotimplyKlc}
Let  $A_i = 0.5$, $B_i =0$, and $C_i = D_i = 1$ for $i = 1, 2$, and $A_0 = 1$. Consider the adjacency matrix $\mathcal{A}$ with $a_{21} = 1$, $a_{12} = 0$, and the pinning gains $g_{1} = 1$ and $g_{2} = 0$.
Choose $G_{1i} = G_{2i} = 1$ for $i = 1, 2$. Take $K_{1i} = -1$ for $i=1,2$, $K_{21} = -0.5$, and $K_{22} = -1$. Then \eqref{eq:lyapunov} is feasible for $P = I_4$, and hence, $K\in K_{\mathrm{S}}$.  Suppose for a contradiction that $K \in K_\mathrm{LC}$. Then there exist $r_1  \geq \sigma_{\mathrm{max}}^3(\mathcal{F}\mathcal{A})/\sigma_\mathrm{min}(\mathcal{F}\mathcal{A}) $ and $P_1 \succ 0$, with the partition in \eqref{eq:small_P_partition}, such that the inequalities \eqref{eq:agentwise_local_inequalities} hold. Note that $\sigma_\mathrm{max}(\mathcal{F}\mathcal{A}) = \sigma_\mathrm{min}(\mathcal{F}\mathcal{A}) = 1$. Thus, $r_1 \geq 1$, and the last two inequalities in \eqref{eq:agentwise_local_inequalities} imply that $P_{21}
= 4$. Since the first inequality in \eqref{eq:agentwise_local_inequalities} is equivalent to the first inequality in \eqref{eq:ARElike}, $M_1 \prec 0$. Thus, $2r_1 -2$, which is a diagonal entry of $M_1$, must be negative.  
In conjunction with $r_1 \geq 1$, this implies that $1<1$, a contradiction. Hence, $K_{\mathrm{S}} \nsubseteq K_{\mathrm{LC}}$. 
Considering this and the first inclusion in Corollary \ref{crl:setinclusion1}, we conclude that $K_{\mathrm{LC}}$ is a proper subset of $K_{\mathrm{S}}$ for this example. 
\end{example}

\begin{example}\label{ex:KadoesnotimplyKg}
Consider Example \ref{ex_structured_stabilizability}. Take
$K = -[I_2 \ I_2]$. Then  $A_{\mathrm{f}i}$ is Schur for $i=1,2$, but $A_\mathrm{g}$ is not.
Thus,  $K_{\mathrm{LA}} \nsubseteq K_{\mathrm{G}}  $.
\end{example}

\begin{example}\label{ex:KGdoesnotimplyKa}
Consider Example \ref{ex:KgdoesnotimplyKs}. Even though $K \in K_{\mathrm{G}}$,   $K \notin K_{\mathrm{LA}}$ since $A_{\mathrm{f}1}$ is not Schur. Thus, $K_{\mathrm{G}} \nsubseteq K_{\mathrm{LA}}$. Note also that Conditions \ref{ass:spanningtree}--\ref{ass:pcopy} hold. It is now worth highlighting that, by Theorem  \ref{thm:keytheorem}, Problem \ref{prb:mainproblem} is still solved without $A_{\mathrm{f}1}$ being Schur.
\end{example}




\begin{corollary}\label{crl:setinclusion2}
If  $\mathcal{G}$ is acyclic, then $K_\mathrm{LA} = K_\mathrm{G}$.
\end{corollary}
\begin{IEEEproof}
   It is immediate from Corollary \ref{crl:acyclic} (i).
\end{IEEEproof}

Finally, the following example demonstrates the agent-wise local synthesis of a $K\in K_{\mathrm{LC}}$ that solves Problem \ref{prb:mainproblem} for  a heterogeneous MAS by utilizing Corollary \ref{crl:ControlSynthesisDiszero}.

\begin{example}
    Let $A_i = B_i = C_i = 1$, $i=1,3$, 
    \begin{align}\nonumber
    A_i &= \begin{bmatrix}
        1 & 1 \\ 0 & 1
    \end{bmatrix}, \ B_i = \begin{bmatrix}
        0.5 \\ 1
    \end{bmatrix}, \ C_i = \begin{bmatrix}
        1 & 0
    \end{bmatrix}, \ \ i=2,4 \nonumber
\end{align}  
$D_i = 0 $, $i = 1, 2,3, 4$, and $A_0 = 1$. Consider the adjacency matrix $\mathcal{A}$ with $a_{21} = a_{12} = a_{32} = 0.2$, $a_{14} = a_{24} = a_{34} = a_{23} = 0.1$, and the pinning gains $g_1 = 0.5$, $g_4 = 0.1$.   The remaining entries of $\mathcal{A}$ and the rest of the pinning gains are zero. 
Select $G_{1i} = G_{2i} = 1$ for $i = 1, 2, 3 ,4$.  Note that Conditions \ref{ass:spanningtree}--\ref{ass:pcopy} hold. Choose $r_i = 0.92 \geq \sigma_{\mathrm{max}}^3(\mathcal{F}\mathcal{A})/\sigma_\mathrm{min}(\mathcal{F}\mathcal{A}) $ for $i=1,2,3,4$. To solve \eqref{eq:convex_general_D_LMI}, we use CVX 
a package for modeling and solving convex programs \cite{grant2014cvx}, and obtain
\begin{align}
  P_i &= \begin{bmatrix}
            0.7231  &  -1.8112 \\
   -1.8112  &  20.7015
    \end{bmatrix} \nonumber \\ 
    Y_i &=  \begin{bmatrix}
        -0.7377 & -0.0532
    \end{bmatrix}, \ \Theta_i = 1.2077, \ i = 1, 3 \nonumber \\ 
    P_i &= \begin{bmatrix}
        0.7033 &  -0.8646 &   -1.7348 \\
   -0.8646  &  6.2609  &   -0.0550 \\
   -1.7348 &  -0.0550 &   17.0840
    \end{bmatrix} \nonumber \\
   Y_i &= \begin{bmatrix}
         0.5105  & -8.4232 &    0.1100
    \end{bmatrix}, \ \Theta_i = 12.8017, \ i= 2,4. \nonumber 
\end{align}
Each follower's control gain $K_i$ is recovered as suggested in Corollary \ref{crl:ControlSynthesisDiszero}, and is given below 
\begin{align}
    K_i &= \begin{bmatrix}
           -1.3147  & -0.1176
    \end{bmatrix}, \ i = 1,3 \nonumber  \\ 
    K_i &= \begin{bmatrix}
         -1.5978  & -1.5674 &   -0.1609
    \end{bmatrix}, \  i = 2,4. \nonumber
\end{align}
By the proof of Corollary \ref{crl:ControlSynthesisDiszero}, $K \in K_{\mathrm{LC}}$. We also conclude from  Corollary \ref{crl:ControlSynthesisDiszero} that  Problem \ref{prb:mainproblem} is solved. 
\end{example}

Although it provides sufficient conditions for the solvability of the problem for a given $K_i$, Theorem \ref{thm:localsolution} does not provide any guide for control synthesis. The following result shows that a convex relaxation of the agent-wise local sufficient condition leads to an independent control synthesis for each agent.

\section{Conclusion}
We studied the solvability of the RCORP of discrete-time uncertain heterogeneous linear MASs through the distributed internal model approach.  Global and agent-wise local design methods, involving the feasibility of LMIs, were presented. 
The global one is less conservative, whereas the local is more scalable. The solvability result and design methods in this article can serve as a building block for research in data-driven distributed control, for instance, in extending the data informativity approach for robust output regulation \cite{10502160} and the data-driven cooperative output regulation \cite{FATTORE2026112789} to RCORP.

\bibliographystyle{IEEEtran}
\bibliography{refs.bib}

\end{document}